\documentclass[runningheads]{llncs}
\usepackage[T1]{fontenc}
\def\doi#1{\href{https://doi.org/\detokenize{#1}}{\url{https://doi.org/\detokenize{#1}}}}
\usepackage{graphicx}
\usepackage{listings}
\usepackage{amsmath}
\usepackage{amssymb}
\usepackage{graphicx}
\usepackage{complexity} 
\usepackage{enumitem}
\setlist{nosep}
\usepackage{algorithm}
\usepackage[noend]{algpseudocode}
\usepackage{footnote}
\algnewcommand{\LeftComment}[1]{\Statex \(\triangleright\) #1}
\newcommand{\blueLcomment}[1]{\LeftComment{ \blue{\small\texttt{#1}}}}

\usepackage[svgnames]{xcolor}
\usepackage{color}

\newcommand{\blue}[1]{{\color{blue}{#1}}}

\newcommand{\calA}{\mathcal{A}}

\newcommand{\calO}{\ensuremath{{\mathcal O}}}

\newcommand{\calV}{\mathcal{V}}

\newcommand{\true}{\texttt{True}}
\newcommand{\false}{\texttt{False}}

\newcommand{\yes}{\textsc{Yes}}
\newcommand{\no}{\textsc{No}}

\newtheorem{observation}[theorem]{Observation}
\newtheorem{reduction-rule}{Reduction Rule}
\newtheorem{branching-rule}{Branching Rule}
\newtheorem{branching-rule1}{Branching Rule}
\newtheorem{reduction rule}{Reduction Rule}
\newtheorem{preprocessing rule}{Preprocessing Rule}
\newtheorem{branching rule}{Branching Rule}
\newtheorem{marking-scheme}{Marking Scheme}

\algrenewcommand\algorithmicrequire{\textbf{Input:}}
\algrenewcommand\algorithmicensure{\textbf{Task:}}

\newcommand{\defproblem}[3]{
  \vspace{1mm}
\noindent\fbox{
  \begin{minipage}{0.96\textwidth}
  \begin{tabular*}{\textwidth}{@{\extracolsep{\fill}}lr} #1 \\ \end{tabular*}
  {\bf{Input:}} #2  \\
  {\bf{Question:}} #3
  \end{minipage}
  }
  \vspace{1mm}
}

\usepackage{hyperref}
\usepackage{thmtools}
\usepackage{thm-restate}

\begin{document}
\title{Revisiting Path Contraction and Cycle Contraction}

\author{R. Krithika~\thanks{The author is supported by SERB MATRICS grant number MTR/2022/000306.}~\inst{1} \and V. K. Kutty Malu \inst{1} \and
Prafullkumar Tale \inst{2}}
\authorrunning{Krithika, Kutty Malu and Tale}
%
\institute{Indian Institute of Technology Palakkad, Palakkad, India \\ \email{krithika@iitpkd.ac.in | 112104002@smail.iitpkd.ac.in}  
\\ 
\and Indian Institute of Science Education and Research Bhopal, Bhopal, India \\
\email{prafullkumar@iiserb.ac.in}}
\maketitle              
\begin{abstract}
The \textsc{Path Contraction} and \textsc{Cycle Contraction} problems take as 
input an undirected graph $G$ with $n$ vertices, $m$ edges and an integer 
$k$ and determine whether one can obtain a path or a cycle, respectively, by 
performing at most $k$ edge contractions in $G$. We revisit these \NP-complete problems and prove the following results.
\begin{itemize}
    \item \textsc{Path Contraction} admits an algorithm running in $\mathcal{O}^*(2^{k})$ time. This improves over the current algorithm known for the problem [Algorithmica 2014].
    \item \textsc{Cycle Contraction} admits an algorithm running in $\mathcal{O}^*((2 + 
\epsilon_{\ell})^k)$ time where $0 < \epsilon_{\ell}  \leq  0.5509$ is inversely proportional to $\ell = n - k$. 
\end{itemize}    
Central to these results is an algorithm for a general variant of \textsc{Path Contraction}, namely, \textsc{Path Contraction With Constrained Ends}. We also give an $\mathcal{O}^*(2.5191^n)$-time algorithm to solve the optimization version of \textsc{Cycle Contraction}. 

Next, we turn our attention to restricted graph classes and show the following results. 
\begin{itemize}
\item \textsc{Path Contraction} on planar graphs
admits a polynomial-time algorithm.
\item \textsc{Path Contraction} on chordal graphs
does not admit an algorithm running in time 
$\mathcal{O}(n^{2-\epsilon} \cdot 2^{o(tw)})$
for any $\epsilon > 0$,
unless the \textsf{Orthogonal Vectors Conjecture} fails.
Here, $tw$ is the treewidth of the input graph.
\end{itemize}
The second result complements the $\mathcal{O}(nm)$-time,
i.e., $\mathcal{O}(n^2 \cdot tw)$-time, algorithm known for the problem [Discret. Appl. Math. 2014].


\end{abstract}

\section{Introduction}
Graph editing problems have consistently been benchmark problems against which the power of new algorithmic tools and techniques are tested. 
Typical editing operations are vertex deletions, edge deletions, edge additions, edge subdivisions and edge contractions. In this work, we focus on editing a simple undirected graph by only performing edge contractions. Contracting an edge results in the addition of a new vertex adjacent to the neighbors of its endpoints followed by the deletion of the endpoints. One may view an edge contraction as ``merging'' its endpoints into one thereby making them indistinguishable. Any loops or parallel edges created in the process are deleted so that the graph remains simple. In graph contraction problems, given a graph $G$ and an integer $k$, the interest is in determining if $G$ can be modified into a graph that belongs to a specific family of graphs by a sequence of at most $k$ edge contractions. Most basic such graph families are paths and cycles. This brings us to the problems studied in this paper, namely, \textsc{Path Contraction} and \textsc{Cycle Contraction}. 


\defproblem{\textsc{Path Contraction}}{A connected undirected graph $G$ and an integer $k$.}{Can one contract at most $k$ edges in $G$ to obtain a path?}

\defproblem{\textsc{Cycle Contraction}}{A connected undirected graph $G$ and an integer $k$.}{Can one contract at most $k$ edges in $G$ to obtain a cycle?}

\noindent Let $P_{\ell}$ and $C_{\ell}$ denote the path and cycle respectively on $\ell$ vertices. When viewed as vertex sets rather than graphs, $P_{\ell}$ and $C_{\ell}$ are referred to as induced $\ell$-path and induced $\ell$-cycle, respectively. Early results by Brouwer and Veldman~\cite{BrouwerV87} show that \textsc{Path Contraction} and \textsc{Cycle Contraction} are \NP-complete\ even when the target graph (which is a path or a cycle) has four vertices (i.e., when $k=n-4$). However, \textsc{Path Contraction} and \textsc{Cycle Contraction} are polynomial-time solvable when the target graph is a $P_{3}$ or a $C_{3}$~\cite{BrouwerV87}. In fact, Brouwer and Veldman~\cite{BrouwerV87} showed that \textsc{Path Contraction} is \NP-complete\  when the target path has $\ell$ vertices for every $\ell \ge 4$. Later, Hammack~\cite{Hammack02} showed that \textsc{Cycle Contraction} is \NP-complete\ when the target cycle has $\ell$ vertices for every $\ell \ge 5$. These results imply that \textsc{Path Contraction} and \textsc{Cycle Contraction} are \NP-complete\ in general.

\textsc{Path Contraction} and \textsc{Tree Contraction} were the first graph contraction problems to be studied in the parameterized complexity framework. The most natural parameter is $k$ and Heggernes et al.~\cite{HeggernesHLLP14} showed that \textsc{Path Contraction} admits a kernel with at most $5k+3$ vertices and an \FPT\ algorithm running in $2^{k + \mathcal{O}(\sqrt{k} \log k)}$ time. They also showed that \textsc{Path Contraction} can be solved in $\mathcal{O}^*(2^n)$ time\footnote[1]{$\mathcal{O}^*(.)$ suppresses polynomial factors} by a simple algorithm -- color the graph with two colors, contract the connected components in the monochromatic subgraphs and check if the resulting graph is a path. Later, Li et al.~\cite{LI201720} improved the size of the kernel to at most $3k+4$ vertices. While there is no improvement to the $2^{k + \mathcal{O}(\sqrt{k} \log k)}$-time algorithm known till date, Agrawal et al.~\cite{AgrawalFLST20} gave an $\mathcal{O}^*(1.99987^n)$-time algorithm  breaking the $\mathcal{O}^*(2^n)$ barrier. 


Our approach to \textsc{Path Contraction} has a different style when compared to these results. We focus on the following variant of \textsc{Path Contraction}. Refer to Section~\ref{sec:prelims} for the definition of witness structure. 

\defproblem{\textsc{Path Contraction With Constrained Ends}}{A connected graph $G$, two disjoint subsets $X,Y \subseteq V(G)$ and an integer $k$.}{Can one contract at most $k$ edges in $G$ to obtain a path with witness structure $(W_1,\ldots,W_\ell)$ such that $X \subseteq W_1$ and $Y \subseteq W_\ell$?}


\noindent This problem is \NP-hard even if $\ell=2$ as it generalizes the \textsc{2-Disjoint Connected Subgraphs} problem which is known to be \NP-hard~\cite{HofPW09}. Our first result is an algorithm for \textsc{Path Contraction With Constrained Ends} using a dynamic programming formulation that crucially uses the notions of \textit{potential $k$-witness sets} (Definition \ref{def:k-potential-witness-set}) and \textit{potential $k$-prefix sets} (Definition \ref{def:k-nice-subset}). It is inspired by the exact exponential-time algorithm by Agrawal et al.~\cite{AgrawalFLST20} for \textsc{Path Contraction} and an \FPT\ algorithm by Saurabh et al.~\cite{SaurabhST20} for \textsc{Grid Contraction}. 

\begin{restatable}[]{thm}{pathEndsfptalgo}
\label{thm:path-ends-fpt-algo}
\textsc{Path Contraction With Constrained Ends} admits an algorithm
running in time $\mathcal{O}^*(2^{k-|X|-|Y|})$.
\end{restatable}

\noindent Observe that \textsc{Path Contraction With Constrained Ends} when $X=\emptyset$ and $Y=\emptyset$ is \textsc{Path Contraction}. Therefore, we obtain an algorithm for \textsc{Path Contraction}.

\begin{corollary}
\label{cor:path-conn-algo}
\textsc{Path Contraction} admits an algorithm
running in time $\mathcal{O}^*(2^{k})$.
\end{corollary}

\noindent This is the first improvement for \textsc{Path Contraction} since its fixed-parameter tractability was established by Heggernes et al.~\cite{HeggernesHLLP14}. 

Our original motivation for focusing on \textsc{Path Contraction With Constrained Ends} was to understand its applicability in designing faster algorithms for \textsc{Cycle Contraction} since deleting one vertex from a cycle results in a path whose endpoints are neighbours of the deleted vertex. As cycles have treewidth $2$, one may use Courcelle’s theorem to show the existence of \FPT\ algorithms for \textsc{Cycle Contraction} (see \cite[Chapter 7]{CyganFKLMPPS15} for related definitions and arguments). However, the resulting algorithm has huge exponential factors in the running time. Belmonte et al.~\cite{BelmonteGHP14} showed that \textsc{Cycle Contraction} admits a kernel with at most $6k+6$ vertices which was later improved to $5k+4$ by Sheng and Sun~\cite{SHENG201914}. This immediately results in a single-exponential time algorithm, i.e., an $\mathcal{O}^*(c^k)$-time algorithm for some constant $c>200$, for \textsc{Cycle Contraction} -- simply color the graph with two or three colors, contract the connected components in the monochromatic subgraphs and check if the resultant graph is a cycle. We describe a faster algorithm using the algorithm for \textsc{Path Contraction With Constrained Ends} as a subroutine.

\begin{restatable}[]{thm}{fptalgocycle}
\label{thm:fpt-algo-cycle}
\textsc{Cycle Contraction} admits an algorithm running in time
$\mathcal{O}^*((2 + \epsilon_{\ell})^k)$ where 
$0 < \epsilon_{\ell}  \leq  0.5509$ and $\epsilon_\ell$ is inversely
proportional to $\ell = n - k$. 
\end{restatable}

\noindent This also implies an exact exponential-time algorithm for \textsc{Cycle Contraction} that runs in $\mathcal{O}^*(2.5509^n)$ time. We further improve the running time to $\mathcal{O}^*(2.5191^n)$ (Theorem \ref{thm:exact-algo-cycle}).

\textsc{Cycle Contraction} is related to an important graph parameter, namely, \textit{cyclicity}. The cyclicity of a graph is the largest integer $\ell$ such that the graph is contractible to $C_\ell$. Computing cyclicity is \NP-hard\ in general~\cite{BrouwerV87}, however, Hammack~\cite{Hammack99} showed that the cyclicity of planar graphs can be computed in polynomial time. Therefore, \textsc{Cycle Contraction} is polynomial-time solvable on planar graphs. To the best of our knowledge, the status of \textsc{Path Contraction} on planar graphs is open. 

\begin{restatable}[]{thm}{pathplanar}
\label{thm:path-planar}
\textsc{Path Contraction} on planar graphs admits a polynomial-time algorithm.
\end{restatable}

\noindent This result uses the \textsc{Cycle Contraction} algorithm on planar graphs as a subroutine to obtain a polynomial-time algorithm for \textsc{Path Contraction}. Note that \textsc{Tree Contraction} is \NP-hard on planar graphs~\cite{AsanoH83}. Hence, planar graphs is an example graph class where \textsc{Path Contraction} is easier than \textsc{Tree Contraction}. 

On chordal graphs, while \textsc{Cycle Contraction} is trivially solvable (in linear time), the status of \textsc{Path Contraction} is interesting. Heggernes et al.~\cite{HeggernesHLP14} proved that \textsc{Path Contraction} is polynomial-time solvable on chordal graphs by giving an $\mathcal{O}(nm)$-time algorithm where $m$ is the number of edges in the input graph. This is also an $\mathcal{O}(n^2 \cdot tw)$-time algorithm where $tw$ is the treewidth of $G$. We prove that this algorithm is optimal in the following sense. 

\begin{restatable}[]{thm}{chordallb}
\label{chordallb}
Unless the \textsf{Orthogonal Vectors Conjecture} fails,
\textsc{Path Contraction} on chordal graphs does not admit
an algorithm running in time 
$\mathcal{O}(n^{2-\epsilon} \cdot 2^{o(tw)})$
for any $\epsilon > 0$.
\end{restatable}

This result specifically shows that the $\mathcal{O}(n^2 \cdot tw)$-time algorithm given by Heggernes et al. is optimal -- if one wishes to reduce the quadratic dependency on $n$, an exponential overhead in terms of treewidth is inevitable. Heggernes et al.~\cite{HeggernesHLP14} also showed that \textsc{Tree Contraction} on chordal graphs admits an $\mathcal{O}(n+m)$-time algorithm. This result combined with Theorem~\ref{chordallb} implies that chordal graphs is an example graph class where \textsc{Tree Contraction} is easier than \textsc{Path Contraction}. 

\section{Preliminaries}
\label{sec:prelims}
For $k \in \mathbb{N}$, $[k]$ denotes the set $\{1,2,\ldots, k\}$. We refer to the book by Diestel~\cite{Diestel12} for standard graph-theoretic definitions and terminology not defined here. For an undirected graph $G$, sets $V(G)$ and $E(G)$ denote its set of vertices and edges, respectively.  For a vertex $v$, $N_G(v)$ denotes the set of its neighbors in $G$. We omit the subscript in the notation for neighborhood if the graph under consideration is clear. For a set $S \subseteq V(G)$, $\overline{S}$ denotes the set $V(G) \setminus S$. For a set $S \subseteq V(G)$, we denote the graph obtained by deleting $S$ from $G$ by $G - S$ and the subgraph of $G$ induced on the set $S$ by $G[S]$. For two subsets $S_1, S_2 \subseteq V(G)$, let $E(S_1,S_2)$ denote the set of edges with with one endpoint in $S_1$ and other in $S_2$ and we say that $S_1$ and $S_2$ are adjacent if $E(S_1,S_2)\neq \emptyset$. A set $S$ of $V(G)$ is said to be a \textit{connected set} if $G[S]$ is connected. We state the following notions of special connected sets given in \cite{AgrawalFLST20}.




\begin{definition}[$(Q,\alpha,\beta)$-connected Set]
\label{def:Q-a-b}
For integers $\alpha,\beta \in \mathbb{N}$ and a non-empty set $Q\subseteq V(G)$, a connected set $X$ is a $(Q,\alpha,\beta)$-connected set if $Q \subseteq X$, $|X| = \alpha$, and $|N(X)|= \beta$.
\end{definition}

\begin{definition}[$(\alpha,\beta)$-connected Set]
\label{def:a-b}
For integers $\alpha,\beta \in \mathbb{N}$, a connected set $X$ is called an $(\alpha,\beta)$-connected set if $|X| = \alpha$ and $|N(X)| = \beta$. 
\end{definition}

Next, we state a result on the enumeration of $(Q,\alpha,\beta)$-connected sets and $(\alpha,\beta)$-connected sets which follows from Lemma~$3.1$ and Lemma~$3.2$ of~\cite{FominV12} (also see Lemma~$2.4$ and Lemma~$2.5$ of~\cite{AgrawalFLST20}). 

\begin{proposition} 
\label{prop:ab-conn-sets} 
Given integers $\alpha,\beta \in \mathbb{N}$, the number of $(\alpha,\beta)$-connected sets in $G$ is at most $2^{\alpha + \beta} \cdot n$ and the set of all $(\alpha,\beta)$-connected sets can be enumerated in $\calO^*(2^{\alpha + \beta})$ time. Further, given a non-empty set $Q \subseteq V(G)$, the number of $(Q,\alpha,\beta)$-connected sets in $G$ is at most $2^{\alpha + \beta - |Q|}$ and the set of all $(Q,\alpha,\beta)$-connected sets can be enumerated in $\calO^*(2^{\alpha + \beta - |Q|})$ time.
\end{proposition}

\smallskip
\noindent \textbf{Edge contractions.} For a graph $G$ and an edge $e = uv \in E(G)$, $G/e$ denote the graph obtained from $G$ by contracting edge $e$. Formally, $V(G/e) = (V(G) \cup \{w\}) \backslash \{u, v\}$ and $E(G/e) = \{xy \mid x,y \in V(G) \setminus \{u, v\}, xy \in E(G)\} \cup \{wx |\ x \in N_G(u) \cup N_G(v)\}$ where $w$ is a new vertex not in $V(G)$. Observe that an edge contraction reduces the number of vertices in the graph by exactly one. For a set $F \subseteq E(G)$, $G/F$ denotes the graph obtained from $G$ sequentially contracting the edges in $F$; $G/F$ is oblivious to the contraction sequence and two distinct contraction sequences result in the same graph. A graph $G$ is said to be \emph{contractible} to the graph $H$ if there is an onto function $\psi: V(G) \rightarrow V(H)$ such that the following properties hold; in that case, we say that $G$ is contractible to $H$ via mapping $\psi$.
\begin{itemize}
\item For any vertex $h \in V(H)$, $G[W(h)]$ is connected and not empty where $W(h) = \{v \in V(G) \mid \psi(v)= h\}$.
\item For any two vertices $h, h’ \in V(H)$, $W(h)$, $W(h’)$ are adjacent in $G$ if and only if $hh’ \in E(H)$.
\end{itemize}
For a vertex $h$ in $H$, set $W(h)$ is called a \emph{witness set} associated with/corresponding to $h$. We define $H$-\emph{witness structure} of $G$, denoted by $\mathcal{W}$, as a collection of all witness sets; formally $\mathcal{W}=\{W(h) \mid h \in V(H)\}$. Note that a witness structure $\mathcal{W}$ is a partition of $V(G)$. If $G$ has a $H$-witness structure, then $H$ can be obtained from $G$ by a sequence of edge contractions. For a fixed $H$-witness structure, let $F$ be the union of edges of spanning trees of all witness sets where by convention, the spanning tree of a singleton set is an empty set. In order to obtain $H$ from $G$ it suffices to contract edges in $F$. We say that $G$ is \emph{$k$-contractible} to $H$ if $|F|$ is at most $k$, i.e., $H$ can be obtained from $G$ by at most $k$ edge contractions. If $G$ is contractible to a $P_\ell$ or $C_\ell$, we denote the corresponding witness structure as a sequence $(W_1,\ldots,W_\ell)$ to emphasize the order of the witness sets. For a cycle witness structure $(W_1,\ldots,W_\ell)$, $W_1$ and $W_\ell$ are regarded as consecutive witness sets just like the other obvious consecutive witness sets.  

\smallskip \noindent \textbf{Parameterized complexity.} A \textit{parameterized problem} is a decision problem in which every instance $I$ is associated with a natural number $k$ called \textit{parameter}. A parameterized problem $\Pi$ is said to be \emph{fixed-parameter tractable} (\FPT) if every instance $(I,k)$ of $\Pi$ can be solved in $f(k)\cdot |I|^{\calO(1)}$ time where $f(\cdot)$ is some computable function whose value depends only on $k$. We say that two instances, $(I, k)$ and $(I’, k’)$, of a parameterized problem $\Pi$ are \emph{equivalent} if $(I, k) \in \Pi$ if and only if $(I’, k’) \in \Pi$. A parameterized problem $\Pi$ admits a kernel of size $g(k)$ (or $g(k)$-kernel) if there is a polynomial-time algorithm (called {\em kernelization algorithm}) which takes as an input $(I,k)$ of $\Pi$, and in time $|I|^{\calO(1)}$ returns an equivalent instance $(I’,k’)$ of $\Pi$ such that $|I’| + k’ \leq g(k)$ where $g(\cdot)$ is a computable function whose value depends only on $k$. For more details on parameterized complexity, we refer the reader to the book by Cygan et al.~\cite{CyganFKLMPPS15}.


\section{Path Contraction}
In this section, we describe an algorithm to solve \textsc{Path Contraction With Constrained Ends}. Consider an instance $(G,X,Y,k)$. Let $n$ denote the number of vertices in $G$.  We begin with the following observation.




\begin{observation}
\label{obs:sumofws-path}
If $G$ is $k$-contractible to $P_{\ell}$ with witness structure $(W_1,\ldots,W_\ell)$, then $\sum_{i=1}^{\ell} |W_i| \leq k + \ell$.  
\end{observation}

Observation \ref{obs:sumofws-path} immediately leads to a $\mathcal{O}^*(2^{k-|X|-|Y|})$-time algorithm for the problem when $\ell$ is a constant.
For a clearer analysis of the algorithm, we distinguish two cases:
when $\ell \le 5$, i.e., $n \le k + 5$ and otherwise.

\begin{lemma}
\label{lem:P4algo}
There is an algorithm that determines if $G$ is $k$-contractible to a path of at most five vertices with the first witness set containing $X$ and the last witness set containing $Y$ in $\mathcal{O}^*(2^{k-|X|-|Y|})$ time. 
\end{lemma}
\begin{proof}
If $G$ is $k$-contractible to a path of at most five vertices 
then $n \leq k+5$ by Observation~\ref{obs:sumofws-path}. Thus, if $n 
>k+5$, then we declare that $G$ is not $k$-contractible to a path on at 
most five vertices. Subsequently, we assume $n \leq k+5$. Suppose $G$ is $k$-contractible to a path $P_{\ell}$ for some $2 \le \ell \le 5$, and $(W_1,\dots, W_{\ell})$ is the corresponding path
witness structure. Consider a $2$-coloring of $V(G)$ such that every witness set is monochromatic and the adjacent witness sets get different colors.
In such a coloring, the first and the last witness set get the same color
when $\ell$ is odd and different colors when $\ell$ is even.
We consider these two cases seperately.

We first compute the set of all colorings of $V(G)$ using two colors such that $X$ and $Y$ are monochromatic and have the same color. 
Then, we iterate over this set, contract each of the connected 
components of the two monochromatic subgraphs into a single vertex and 
check if the resultant graph is a path with first witness set 
containing $X$ and the last witness set containing $Y$ and at most $k$ edges were contracted. 
We repeat this procedure with the set of all colorings of $V(G)$ using two colors such that $X$ and $Y$ are monochromatic but have different colors.
If no coloring results in the required $P_\ell$ with $\ell \leq 5$, then $G$ is 
not $k$-contractible to a path of at most five vertices with first witness set 
containing $X$ and the last witness set containing $Y$. Otherwise, we obtain 
a $P_\ell$ from $G$ using at most $k$ edge contractions with witness structure 
$(W_1,W_2,\ldots,W_\ell)$ 
where $\ell \leq 5$ and $X \subseteq W_1$, $Y 
\subseteq W_\ell$.  

As there are $\mathcal{O}(2^{k+5-|X|-|Y|}$ such colorings, the overall 
running time of this algorithm is $\mathcal{O}^*(2^{k-|X|-|Y|})$. \qed
\end{proof} 

Subsequently, we consider the case $\ell \geq 6$, i.e. $k \leq n-6$. 
We critically use the fact that any path witness structure of $G$ has at least $6$
witness sets to bound the cardinality of the closed neighborhood of 
a witness set.

\begin{observation}
\label{obs:k-potential-witness-set-prop}
Suppose $G$ is $k$-contractible to $P_\ell$ with $\ell \geq 6$
and $(W_1,W_2,\ldots,W_\ell)$ is the corresponding witness structure with 
$X \subseteq W_1$ and $Y \subseteq W_\ell$. 
\begin{enumerate}
\item For any $2 \le i \le \ell - 1$,
$G - W_i$ has two connected components
$C_X, C_Y$ containing $X$ and $Y$, respectively. 
\item  For any $2 \le i \le \ell - 1$,
$|(N(W_i) \cap W_{i-1}) \setminus X | +  |W_i| + 
|(N(W_i) \cap W_{i+1}) \setminus Y| \le k + 5 - |X|-|Y|$.
\item For any $2 \le i \le \ell - 2$, 
$|(N(W_i) \cap W_{i-1}) \setminus X| + |W_i| + |W_{i+1}| + |(N(W_{i+1}) \setminus W_i)\setminus Y| \le k + 6-|X|-|Y|$.
\end{enumerate}
\end{observation}

The first two properties are about a witness set
whereas the third property is about a pair of consecutive witness sets.
We use the second property to enumerate all potential  
$k$-witness sets and the corresponding potential $k$-prefix sets.
Triples consisting of a potential  
$k$-witness set, a corresponding potential $k$-prefix set and an integer in $[k] \cup \{0\}$ 
will describe a state of the dynamic programming table.
The third point in the observation is used to bound
the time needed for the algorithm to update all the entries in this table.

\subsection{Potential Witness Sets and Potential Prefix Sets}

Intuitively, if $G$ is $k$-contractible to $P_\ell$ with the witness structure 
$(W_1,W_2,\ldots,W_\ell)$ satisfying the required properties, 
a potential $k$-witness set is a candidate for $W_i$ with 
$2 \leq i \leq \ell-1$ and a potential $k$-prefix set is a 
candidate for the set of vertices in a `prefix witness set'.  See Figure~\ref{fig:path} for an illustration. 

\begin{definition}[Potential $k$-witness Set]
\label{def:k-potential-witness-set} 
A connected set $W \subseteq V(G) \setminus (X \cup Y)$ is called a potential $k$-witness set of $G$ if 
$|(N(W) \setminus (X \cup Y) | +  |W|  \le k + 5 - |X|-|Y|$ and $G - W$ has two connected components
$C_X, C_Y$ containing $X$ and $Y$, respectively.
\end{definition}

Note that every non-terminal witness set in a $P_\ell$-witness structure of $G$ is a potential $k$-witness set. Next, we define the notion of a \emph{potential $k$-prefix set}. 

\begin{definition}[Potential $k$-prefix set]
\label{def:k-nice-subset}
For a potential $k$-witness set $W \subseteq V(G) \setminus (X \cup Y)$ with $C_X$ and $C_Y$ being the two connected components of $G - W$ containing $X$ and $Y$, respectively, the sets $S_X := C_X \cup W$ and $S_Y := C_Y \cup W$ are called the \emph{potential $k$-prefix sets} associated with $W$.
\end{definition}

Note that a potential $k$-prefix set is a non-empty connected set $S \subseteq V(G)$ that contains a potential $k$-witness set $W$ such that $G[S]- W$ is one of the connected components of $G - W$. 
The number of potential $k$-prefix sets is twice the number of potential $k$-witness sets since every potential $k$-witness set is associated with two $k$-prefix sets. 
However, we note that a $k$-prefix set $S$  
can be associated with more than one potential $k$-witness set $W$.
Formally, there may exist potential $k$-witness sets $W \neq W'$ such 
that $S = C \cup W = C' \cup W'$,
where $C$ and $C'$ are the connected components of
$G - W$ and $G - W'$, respectively.

We will describe a dynamic programming algorithm in which every entry of the dynamic programming table that we subsequently construct will be indexed by a triple consisting of a potential witness set, a potential prefix set and a non-negative integer that is at most $k$. Therefore, our algorithm's running time will crucially depend upon an efficient enumeration of all potential witness sets and all potential prefix sets. Targeting towards this enumeration, we have the following lemma.

\begin{lemma}
\label{lem:conn-sets} 
For given integers $\alpha, \beta$, the number of connected sets $W \subseteq V(G)\setminus (X \cup Y)$ in $G$ such that $|W|=\alpha$ and $|N(W) \setminus (X \cup Y)|=\beta$ is $\calO^*(2^{\alpha+\beta})$ and these sets can be enumerated in $\calO^*(2^{\alpha+\beta})$ time. 
\end{lemma}
\begin{proof}
    We use a simple depth-bounded search tree algorithm for enumerating the required sets. First we guess a vertex $v \in V(G) \setminus (X \cup Y)$ with the interpretation that $v$ is in the connected set $W$ with the desired properties. The root of the search tree is labelled as $(\{v\},\emptyset)$ and in general each node of the search tree is labelled with a pair $(Z,Z')$ where $|Z| \leq \alpha'$ and $|Z'| \leq \beta'$ for some $\alpha', \beta'$ such that $\alpha' \leq \alpha$ and $\beta' \leq \beta$. With every node $N$ of the search tree with label $(Z,Z')$, we associate the measure $\mu(N)=\alpha+\beta-(|Z|+|Z'|)$. At each node $N$ labelled by $(Z,Z')$ with $\mu(N)>0$, we choose a vertex $u \in N(Z) \setminus (Z' \cup X \cup Y)$ and branch into two possibilities -- $u \in Z$ or $u \in Z'$ resulting in two new nodes with labels $(Z \cup \{u\}, Z')$ and $(Z,Z' \cup \{u\})$ respectively. This branching is exhaustive as the first branch considers the case when $u \in W$ and the second branch considers the case when $u \in N(W) \setminus (X \cup Y)$. As the measure drops by one in each of the two branches, the depth of the tree is at most $\alpha+\beta$ and the number of leaves is at most $2^{\alpha + \beta}$. The leaves $(Z,Z')$ where $Z' = N(Z) \setminus (X \cup Y)$, $|Z|=\alpha$, $|Z'|=\beta$ correspond to the required sets $W=Z$ such that $v \in W$. Now, we repeat this algorithm for each possible initial choice of $v$. Thus, the total number of connected sets $W$ with the required properties is at most $2^{\alpha + \beta} \cdot n$. Further, the set of all such sets are enumerated in $\mathcal{O}^*(2^{\alpha + \beta})$ time.\qed
\end{proof}

As every potential $k$-witness set is also a connected set $W \subseteq V(G)\setminus (X \cup Y)$ such that $|W|=\alpha$, $|N(W) \setminus (X \cup Y)|=\beta$ and $\alpha+\beta \leq k+5-|X|-|Y|$, we have the following corollary. 

\begin{corollary}
\label{cor:pot-int-sets}
    The number of potential $k$-witness sets in $G$ is $\calO^*(2^{k-|X|-|Y|})$ and these sets can be enumerated in $\calO^*(2^{k-|X|-|Y|})$ time.
\end{corollary}


\subsection{A Dynamic Programming Algorithm}

Towards defining the entries in the dynamic programming table $\Gamma$, we define the notion of a \emph{valid index}.
Informally, it is a tuple consisting of a potential $k$-witness set $W$, a potential $k$-prefix set $S$ associated with $W$, and an integer $k' \in \{0\} \cup [k]$.
The value corresponding to this index is set to \true,
if $G[S]$ is $k'$-contractible to a path (with desired properties) 
and $W$ is one of its end witness set. Note that there are two types of potential $k$-prefix sets,
one containing $X$ and the other containing $Y$.
We differentiate between these two types by adding a subscript, i.e.,
$S_X$ denotes a potential $k$-prefix set containing $X$
and $S_Y$ denotes a potential $k$-prefix set containing $Y$. 

\begin{definition}[Valid Index]
\label{def:valid}
For sets $S_X, W \subseteq V(G)$ and $k' \in \{0\} \cup [k]$, the triple $(S_X,W,k')$ is a \emph{valid index} if 
\begin{itemize}
\item $W$ is a potential $k$-witness set,
\item $S_X$ is a potential $k$-prefix set associated with $W$, and
\item $k' \in \{0\} \cup [k]$ and $k' + |N(S_X)| -1 \le k$. 
\end{itemize}
Similarly, we define valid indices corresponding to $S_Y, W \subseteq V(G)$ and 
$k' \in \{0\} \cup [k]$.
\end{definition}

\noindent The table $\Gamma$ has entries only corresponding to valid indices.  
We define 
\begin{align*}
\Gamma[S_X, W, k']  = &\  \true\ \text{if and only if $G[S_X]$ is $k'$-contractible to a path with witness}\\
& \text{structure $(W_1,\ldots,W_q=W)$ for some $q \ge 2$ such that $X \subseteq W_1$}
\end{align*}
\noindent and 
\begin{align*}
\Gamma[S_Y, W, k'']  = &\  \true\ \text{if and only if $G[S_Y]$ is $k''$-contractible to a path with witness}\\
& \text{structure $(W'_1,\ldots,W'_p=W)$ for some $p \ge 2$ such that $Y \subseteq W'_1$}
\end{align*}

Note that the third condition in Definition~\ref{def:valid} insists that $k'$ and $k''$ are small enough to accommodate `future' edge contractions which are at least $|N(S_X)|-1$ or $|N(S_Y)|-1$ in number, respectively.  
Next we describe the initialization of entries of $\Gamma$.
We describe the algorithm for entries of type $\Gamma(S_X, W, k')$
and a similar process is applicable for entries of type $\Gamma(S_Y, W, k')$. 
First, all entries are set to \false. 
Then, the algorithm uses Corollary~\ref{cor:pot-int-sets} 
to enumerate all potential $k$-witness sets
of a predefined size and neighbourhood size 
(which we specify later).
In this step, the algorithm enumerates all the subsets 
that are potential candidates for either 
$W_2$ or $W_{\ell - 1}$.
For each such set $W$, let $S_X$ and $S_Y$ be the two
potential $k$-prefix sets associated with $W$.
The algorithm sets $\Gamma[S_X, W, k'] = \true$
for every $|S_X| - 2  \le k' \le k -|N(S_X)|+1$.
Similarly, it sets $\Gamma[S_Y, W, k'] = \true$
for every $|S_Y| - 2  \le k' \le k-|N(S_Y)|+1$.
The correctness of this initialization follows from the definition.

We now describe how the algorithm updates the entries in $\Gamma$.
Unlike most of the dynamic programming algorithms, our algorithm does not 
`look back' previously set $\true$ entries to update an entry that is set to \false\ in the initialization phase.
Instead, it 
checks entries $\Gamma[S_X, W, k']$ that are $\true$ and `looks forward' 
and sets certain entries $\Gamma[S^{\circ}_X, A, k^{\circ}]$ to \true\ 
where $S_X \subseteq S^{\circ}_X$, 
$(N(W) \setminus S_X) \subseteq A$, and ${k' + |A| - 1} \le k^{\circ} \leq k - |N_{G-S_X}(A)|+1$.
Intuitively, $A$ is a subset that can be appended to the 
path witness structure of $G[S_X]$ as a new witness set to obtain a path witness structure of $G[S_X \cup A]$. 
Hence, if we can obtain a path from $G[S_X]$ (with desired properties) by 
contracting at most $r$ edges, then we can obtain a path from 
$G[S_X \cup A]$ (with desired properties)
by contracting at most $r + |A| - 1$ edges. Thus, $\Gamma[S^{\circ}_X, A, 
k^{\circ}]$ can be set to \true. 
Refer to Figure~\ref{fig:path} for an illustration. 

\begin{figure}[t]
\centering
\includegraphics[scale=.75]{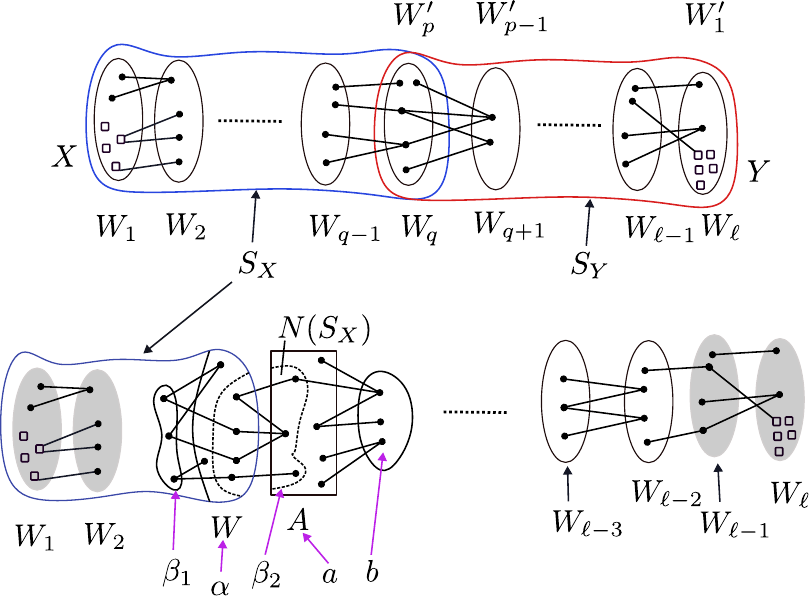}
\caption{(Top)
$W_q$ is one of the potential $k$-witness sets and $S_X$ and $S_Y$ are the corresponding potential $k$-prefix sets. The objective is to determine whether $G[S_X]$ and $G[S_Y]$ admit path witness structures $(W_1, W_2, \dots, W_{q})$ and $(W'_1, W'_2, \dots, W'_p)$ such that 
$X \subseteq W_1$, $Y \subseteq W'_1$, $W_q = W'_p$ and total number of edges contracted is at most $k$ 
(Bottom) 
$W_1$, $W_2$, $W_\ell$, $W_{\ell-1}$ (shaded sets) are identified during the initialization phase.
For a potential $k$-witness set $W$ and the corresponding 
potential $k$-prefix set $S_X$, the algorithm enumerates 
all the possible sets $A$ that can be appended to a path witness structure of $G[S_X]$. 
The witness set $W_q$, mentioned in the top figure, lies
in the unshaded area. The idea is to expand $S_X$ until it contains $W_{\ell - 2}$. The sets with sizes $\beta_1, \alpha, a$ and $b$ corresponds to four internal witness sets (except in some corner cases) justifying the search for sets $A$ satisfying $\beta_1 + \alpha + a + b \le k + 6 - |X|-|Y|$.  
}
\label{fig:path}
\end{figure}

This brings us to the following notion of {\em potential $(S_X,W)$-attachments}. 


\begin{definition}[Potential $(S_X,W)$-attachment]
\label{def:pot-att}
For a potential $k$-witness set $W$ and the potential $k$-prefix set 
$S_X$ associated with $W$, 
the set $\calA(S_X, W)$ of all potential $(S_X, W)$-attachments is defined as $\bigcup_{a=|N(S_X)|}^{k+1} \calA_{a}(S_X, W)$ where 
\begin{align*}
\calA_{a}(S_X, W) = & \{ A\subseteq V(G) \setminus S_X \mid G[A] 
\text{ is connected, } N(S_X) \subseteq A, |A| = a, \text{ and } \\
& |(N(W) \cap S_X) \setminus X| + |W| + |A| + |(N(A) \setminus S_X)\setminus Y| \\
& \hspace{1cm} \le k + 6-|X|-|Y| \}
\end{align*}
\end{definition}


We now justify the upper bound on the cardinality of sets mentioned in the definition. Suppose $G$ is $k$-contractible to $P_\ell$ with witness structure \sloppy $(W_1,\dots, W_{i-1}, W_i,
W_{i+1}, W_{i+2}, \dots, W_{\ell})$ where 
$X \subseteq W_1$ and $Y \subseteq W_{\ell}$.
In the initialization phase, the algorithm identifies sets $W_1, W_2$
and $W_{\ell - 1}, W_{\ell}$. 
Hence, a set $A$ in the potential $(S_X,W)$-attachment 
$\calA_{a}(S_X, W)$ corresponds to the case
when $W$ is some internal witness set $W_i$ 
where $2 \le i \le \ell - 2$.
In fact, we can improve the upper bound to $\ell - 3$ and our algorithm tries to expand 
$S_X$ only till $W = W_{\ell - 3}$ to find 
$A$ which is a potential candidate for $W_{\ell - 2}$.
This implies that $W_1, W_{i-1}, W_{i}, W_{i+1}, W_{i+2}, W_{\ell}$
are six different witness sets when $2 < i \le \ell - 3$ and
are five different sets when $i = 2$ since $6 \le\ell$.

By the definitions of the relevant sets, when $W = W_i$
we have $S_X = W_1 \cup W_2 \cup \cdots \cup W_{i}$.
Hence, $N(W) \cap S_X$ is part of $W_{i-1}$,
$A$ is a potential candidate for $W_{i+1}$  and $N(A) \setminus S_X$
is part of $W_{i+2}$.
As $X$ and $Y$ are part of $W_1$ and $W_\ell$, respectively,
for $2 < i \le \ell - 3$ we have
\begin{align}
|W_1| + |W_{i-1}| + |W_{i}| + |W_{i+1}| + |W_{i+2}| + |W_{\ell}| & \le k + 6
\nonumber \\ 
|X| + |N(W) \cap S_X| + |W| + |A| + |N(A) \setminus S_X| + |Y| & \le k + 6 \nonumber\\
|N(W) \cap S_X| + |W| + |A| + |N(A) \setminus S_X| & \le k + 6 - |X| - |Y|
\label{eq:expand-cardinality-bound}
\end{align}
For $i = 2$, we have $W_1 = W_{i-1}$ and hence
\begin{align}
|(N(W) \cap S_X)\setminus X| + |W| + |A| + |N(A) \setminus S_X| & \le k + 5 - |X| - |Y|
\label{eq:expand-cardinality-bound-2}
\end{align}

Note that Equation \ref{eq:expand-cardinality-bound} subsumes Equation \ref{eq:expand-cardinality-bound-2}. Coming back to the updating phase of the algorithm,
it processes each entry $\Gamma[S_X,W,k']$ in the increasing order of $|S_X|$, $|W|$ and $k'$.
If $\Gamma[S_X,W,k']=\false$, then it skips to the next entry.
Otherwise, it compute $\calA(S_X, W)$ 
(we specify the steps in this computation later).
Then, for every set $A$ in $\calA(S_X, W)$ and
every $k' + (|A| - 1) \le k^{\circ}  \le k - |N_{G-S_X}(A)|+1$,
it sets $\Gamma[S_X \cup A, A, k^{\circ}] =\true$.
This completes the updating phase of the algorithm.

In the final phase, the algorithm iterates over all
potential $k$-witness sets $W$ and the corresponding
valid entries.
Suppose $S_X$ and $S_Y$ are the two potential $k$-prefix sets
associated with $W$.
If there is pair of integers $k_1, k_2$ such that
$k_1 + k_2 - (|W| - 1) \le k$, and
$\Gamma[S_X, W, k_1] = \true$ and
$\Gamma[S_Y, W, k_2] = \true$,
then the algorithm concludes that $G$ is $k$-contractible
to a path with desired properties.
This concludes the description of the algorithm
which we summarize as Algorithm~\ref{alg:path-conn}.

\begin{algorithm}[t]
	\caption{\textsf{PC-DP$(G, X, Y, k)$}}
	\label{alg:path-conn}
	\begin{algorithmic}[1]
		\Require An undirected graph $G$ on $n$ vertices, disjoint subsets $X, Y \subseteq V(G)$ and a positive integer $k$ such that $n> k+5$.
		\Ensure Determine if $G$ is $k$-contractible to a path such that $X$ is contained in one of the end witness sets and $Y$ is contained in the other end witness set	
		\blueLcomment{Phase 1: Initializing $\Gamma$}
		\State Construct a table $\Gamma$ whose entries are index by valid indices 
        \State Initialize all entries in $\Gamma$ to \false 
		\State Enumerate the set $\mathcal{W}$ of all potential $k$-witness sets using Corollary~\ref{cor:pot-int-sets}    
		\For {each $W \in \mathcal{W}$}
		\State Let $S_X$ and $S_Y$ be the $k$-prefix sets associated with $W$
		\State Set $\Gamma[S_X, W, k'] = \true$ for every $|S_X|-2\leq k' \leq k-|N(S_X)|+1$
		\State Set $\Gamma[S_Y, W, k'] = \true$ for every $|S_Y|-2\leq k' \leq k-|N(S_Y)|+1$
        \EndFor		
        \blueLcomment{Phase 2: Updating $\Gamma$}
        \For {each entry $\Gamma[S,W,k']$ in the increasing order of $|S|$, $|W|$, $k'$}
        	\If{$\Gamma[S,W,k']=\false$}
        		\State Skip to the next entry in $\Gamma$
        	\Else 
        		\State Let $|W|=\alpha$, $|(N(W) \cap S) \setminus (X \cup Y)| = \beta_1$ and $|N(W) \setminus S|=\beta_2$.
        		\For{every integer $\beta_2 \leq a \leq k+1$}
        			\State Compute $\calA_{a}(S, W)$ using Proposition~\ref{prop:ab-conn-sets}
        		\EndFor
        		\State Compute $\calA(S, W)=\bigcup_{a=|N(S|}^{k+1} \calA_{a}(S, W)$ 
                \For{every set $A$ in $\calA(S, W)$}
					\For{every $k' + |A| - 1 \leq k^{\circ} \leq k-|N_{G-S}(A)|+1$ }       			
        				\State Set $\Gamma[S \cup A, A, k^{\circ}] =\true$ 
        			\EndFor
        		\EndFor
			\EndIf
		\EndFor
	\blueLcomment{Phase 3: Returning the output}
	\For {each $W \in \mathcal{W}$}
		\State Let $S_X$ and $S_Y$ be the $k$-prefix sets associated with $W$
		\If{$\Gamma[S_X, W, k_1] = \Gamma[S_Y, W, k_2]=\true$ for some $k_1+k_2-|W|+1 \leq k$}
			\State \Return $(G,X,Y,k)$ is a \yes-instance
		\EndIf
	\EndFor
	\State \Return $(G,X,Y,k)$ is a \no-instance
	\end{algorithmic}
\end{algorithm}


\subsection{Correctness and Running Time Analysis}
In the following lemma, we justify the steps 
that updates the value of entries in the table.
Once again, we argue correctness only for entries of type $\Gamma[S_X, W, k']$
as the argument is identical for entries of type $\Gamma[S_Y, W, k']$.

\begin{lemma}\label{lemma:table-correct}
Algorithm~\ref{alg:path-conn} correctly computes 
$\Gamma[S_X, W, k']$ for every valid index $(S_X, W,k')$. 
\end{lemma}
\begin{proof}
We prove that $\Gamma[S_X, W, k']$ is set to 
$\true$ if and only if $G[S_X]$ is 
$k'$-contractible to a path with witness structure 
$(W_1,W_2,\dots, W_q)$ such that $X \subseteq W_1$ and 
$W_q= W$. 

The correctness of the initialization step
follows from the fact that Corollary~\ref{cor:pot-int-sets}
enumerates all $k$-potential witness sets.
Also, once an entry is set to \true, it never changes its value.
Hence, the algorithm sets $\Gamma[S_X, W, k'] = \true$
in the initialization step, if and only if
$G[S_X]$ is $k'$-contractible to a path (that is $P_2$)
with witness structure $(W_1,W_2)$ such that 
$X \subseteq W_1$ and $W_2= W$. 

We prove the correctness of the update phase by induction on $|S_X| + k'$.
The base case corresponds to all the values updated in 
the initialization phase.
Suppose the algorithm 
sets $\Gamma[S_X, W, k']$ correctly for every valid index 
$(S_X,W,k')$ where 
$|S_X| + k' \leq r$ and $r>0$. 
Consider an entry $\Gamma[S^{\circ}_X, W^{\circ}, 
k^{\circ}]$ where $|S^{\circ}_X| + k^{\circ} = r + 1$ and $k^{\circ} \leq k-|N(S^{\circ}_X)|+1$. 
We show that 
$\Gamma[S^{\circ}_X, 
W^{\circ}_X, k^{\circ}]$ is set to $\true$ if and only if 
$G[S^{\circ}_X]$ is $k^{\circ}$-contractible to a path with witness structure 
$(W_1,W_2,\dots, W_q)$ such that $X \subseteq W_1$ and 
$W_q = W^{\circ}$. 

$(\Rightarrow)$ In the forward direction, suppose the algorithm 
sets $\Gamma[S^{\circ}_X, W^{\circ}, k^{\circ}]$ to \true.
If $\Gamma[S^{\circ}_X, W^{\circ}, k^{\circ}]$ is set to $\true$ 
in the initialization phase, then it is the correct value, as mentioned above. 
Consider the case when $\Gamma[S^{\circ}_X, W^{\circ}, k^{\circ}]$ 
is changed from \false\ to \true\ during the update phase. 
This update was done during the inspection of some \true\ entry $\Gamma[\widehat{S}_X, \widehat{W}, \widehat{k}]$. 
Then, the following holds.
\begin{itemize}
\item $W^{\circ} \in \calA_{a}(\widehat{S}_X, \widehat{W})$ where $|N(\widehat{S}_X)| \leq a=|W^{\circ}| \leq k+1$.
\item $S^{\circ}_X=\widehat{S}_X \cup W^{\circ}$ and $\widehat{S}_X \cap W^{\circ} = \emptyset$.
\item $\widehat{k}+a-1 \leq k^{\circ} \leq k-|N(S^{\circ}_X)|+1$. 
\end{itemize}
It follows that $|\widehat{S}_X| = |S^{\circ}_X| - a$ and $k^{\circ} - a + 1 \ge \widehat{k}$.
As $a \geq 1$, we have $|\widehat{S}_X| + \widehat{k} \le  (|S^{\circ}_X| - a) + (k^{\circ} - a + 1) = r+2(1-a) \le r$. 
By induction hypothesis, $\Gamma[\widehat{S}_X, \widehat{W}, \widehat{k}]$ is correctly computed. 
That is, $G[\widehat{S}_X]$ is $\widehat{k}$-contractible to a 
path with witness structure $(W'_{1}, W'_{2}, \dots, W'_{p})$ such that
$X \subseteq W'_1$ and $W'_{p} = \widehat{W}$. 
Then, $(W'_{1}, W'_{2}, \dots, W'_{p}, W^{\circ})$ is a path 
witness structure of $G[S^{\circ}_X]$ implying that $G[S^{\circ}_X]$ 
is $k^{\circ}$-contractible to a path with desired property. 
This concludes the proof for the forward direction of the claim.

$(\Leftarrow)$
In the backward direction, suppose $G[S^{\circ}_X]$ is 
$k^{\circ}$-contractible 
to a path with witness structure $(W_1,W_2,\dots, {W_q})$ such that 
$X \subseteq W_1$ and $W_{q} = W^{\circ}$. 
We show that the algorithm sets 
$\Gamma[S^{\circ}_X, W^{\circ}, k^{\circ}]$ to $\true$. 
If $q = 2$, then because of the correctness of the initialization phase, $\Gamma[S^{\circ}_X, W^{\circ}, k^{\circ}]$ is set to $\true$. 
Consider the case when $q \ge 3$. 
Let $S_X = W_1 \cup W_2 \cup \cdots \cup W_{q-1}$
and $k' = k^{\circ} - (|W_{q}| - 1)$. 
Then, $G[S_X]$ is $k'$-contractible to a path 
with witness structure $(W_1,W_2,\dots, W_{q-1})$
such that $X \subseteq W_1$.
We prove that the algorithm correctly sets 
$\Gamma[S_X, W_{q-1}, k']$ to \true\ and 
then show that while examining $\Gamma[S_X, W_{q-1}, k']$, 
it sets $\Gamma[S^{\circ}_X, W^{\circ}, k^{\circ}]$ to $\true$. 

We begin by showing that the triples $(S_X,W_{q-1},k')$ and 
$(S^{\circ}_X, W^{\circ}, k^{\circ})$ are valid indices. 
Clearly, $W_{q-1}$ is a potential $k$-witness set and 
$G[S_X]-W_{q-1}$ is one of the connected components of 
$G-W_{q-1}$ that contains $X$.
Further, $k' + |N(S_X)|-1 \leq k^{\circ} - (|W_{q}| - 1) + |W_q|-1 \leq 
k-|N(S_X^{\circ})|+1$. 
Hence $(S_X,W_{q-1},k')$ is a valid index. 
Similarly, $W_{q}$ is a potential $k$-witness set and 
$G[S^{\circ}_X]-W_{q}$ is one of the connected components of $G-W_{q}$.
Thus, by this property and the assumption on $k^{\circ}$, it 
follows that $(S^{\circ}_X, W^{\circ}, k^{\circ})$ is a valid index. 
Since $W_q \neq \emptyset$, we have $|S_X| + k' \le |S^{\circ}_X| - |W_q| + k^{\circ} - (|W_q| - 1) \le r + 2 \cdot (1 - |W_q|) \le r$.
By induction hypothesis, we can assume that the algorithm 
correctly sets $\Gamma[S_X, W_{q - 1}, k']$ to \true. 

While examining $\Gamma[S_X, W_{q-1}, k']$, 
the algorithm constructs $\calA_{a}(S_X, W_{q-1})$ for every 
$|N(S_X)| \le a \le k + 1$.
To complete the argument, it suffices to prove that $W_q \in \calA_{a}(S_X, W_{q-1})$ where $a = |W_q|$. 
It is easy to verify that
$W_q \subseteq V(G) \setminus S_X$,
$N(S_X) \subseteq W_q$, and 
$W_q$ is a connected set in $G$. 
It remains to show that condition on the cardinality
mentioned in the definition of $\calA_{a}$ is satisfied.
As mentioned before, the algorithm tries to expand 
$S_X$ till it includes all the witness sets before 
the penultimate witness set $W_{\ell - 1}$, i.e.,
till $S_X = W_1 \cup W_2 \cup \cdots \cup W_{\ell - 2}$.
Hence, when the algorithm 
is enumerating the possible sets $A$ (which are candidates for $W_q$)
to expand $S_X$ that has $W_{q-1}$ as its last witness set,
it is safe to assume that $q - 1 \le \ell - 3$.
Now, by applying Equation~\ref{eq:expand-cardinality-bound} for
$i = q - 1$, we get the desired bound.
Hence, $W_q \in \calA_{a}(S_X, W_{q-1})$, and 
algorithm sets $\Gamma[S^{\circ}_X, W^{\circ}, k^{\circ}]$ to \true.
This concludes the proof of the backward direction
and hence that of the lemma. \qed
\end{proof}

In the next lemma, we argue that the algorithm terminates in the desired time.

\begin{lemma}\label{lemma:table-rt}
Algorithm~\ref{alg:path-conn} terminates in $\calO^*(2^{k-|X|-|Y|})$ time. 
\end{lemma}
\begin{proof}
By Corollary~\ref{cor:pot-int-sets}, the number of potential 
$k$-witness sets is $\calO^*(2^{k - |X| - |Y|})$.
As each potential $k$-witness set is associated with at most two potential 
$k$-prefix sets, it follows that the number of valid indices 
(corresponding to entries in $\Gamma$) is  $\calO^*(2^{k - |X| - |Y|})$. 
Thus, Phase 1 and Phase 3 of the algorithm take $\calO^*(2^{k - |X| - |Y|})$ 
time. 

To bound the running time of Phase 2, we define $\calV_k$ to be the set of all pairs $(S, W)$ such that $S$ is a potential $k$-prefix set and $W$ is a potential $k$-witness set associated with $S$. We partition $\calV_k$ by the size of neighborhood of the potential $k$-witness set in the pair.
Consider three non-negative integers $\alpha, \beta_1, \beta_2$.
We specify the range (i.e., upper bound) of these integers later. 
Define $\calV_k^{\alpha, \beta_1, \beta_2}$ 
as the collection of pairs $(S, W)$ in $\calV_k$ 
such that $W \subseteq V(G)\setminus (X \cup Y)$ is a connected set in $G$ such that $|W|=\alpha$ and $|N(W) \setminus (X \cup Y)|=\beta_1+\beta_2$ with $|(N(W) \cap S) \setminus X| = \beta_1$ 
and $|N(W) \setminus S|=\beta_2$. From the upper bound on the number of such sets given by 
Lemma~\ref{lem:conn-sets} and the fact that every potential $k$-witness set is associated with at most two potential $k$-prefix sets, $|\calV_k^{\alpha,\beta_1, \beta_2}|$ is $\calO^*(2^{\alpha + \beta_1 + \beta_2})$. For every $(S, W) \in  \calV_k^{\alpha,\beta_1,\beta_2}$, 
the algorithm constructs $\calA_a(S,W)$ for every 
$\beta_2 \le a \le k + 1$. 
Define $Q := N(W) \setminus S$.
Then, every set $A$ in $\calA_{a}(S, W)$ is a 
$(Q, a, b)$-connected set in $G - S$ for some 
appropriate value of $b$. 
Thus, we can enumerate elements $A$ of 
$\calA_a[(S, W)]$ with $|N(A)\setminus S|=b$ 
using Proposition~\ref{prop:ab-conn-sets} in 
$\calO^*(2^{a + b - \beta_2})$ time. 
Therefore, the running time of Phase 2 ignoring polynomial factors is 
$$\sum\limits_{\substack{\alpha, \beta_1,\beta_2, a, b}} 2^{\alpha + \beta_1 + \beta_2} \cdot 2^{a + b - \beta_2} = \sum\limits_{\substack{\alpha, \beta_1, a, b}} 2^{\beta_1 +\alpha  + a + b}$$ 
Recall Equation~\ref{eq:expand-cardinality-bound} which 
encodes the cardinality conditions mentioned in the definition 
of $\calA_a[(S, W)]$.
Note that $\beta_1, \alpha, a, b$ corresponds to 
$|N(W) \cap S|, |W|, |A|$,  and $|N(A) \setminus S|$,
respectively. 
See Figure~\ref{fig:path} for an illustration.
Recall that $\calA_a[(S, W)]$ contains only sets for which
the corresponding values satisfy 
$\beta_1 +\alpha  + a + b \le k + 6 - |X| - |Y|$.
Hence, the running time of Phase 2 of the algorithm,
and hence that of the overall algorithm is $\calO^*(2^{k-|X|-|Y|})$. \qed
\end{proof}

This establishes Theorem \ref{thm:path-ends-fpt-algo}. As \textsc{Path Contraction} is \textsc{Path Contraction With Constrained Ends} when $X=\emptyset$ and $Y=\emptyset$, we obtain Corollary \ref{cor:path-conn-algo}. 





\section{Cycle Contraction}
In this section, we give improved \FPT\ and exponential-time algorithms for \textsc{Cycle Contraction}. Let $n$ denote the number of vertices in $G$.

\subsection{An FPT Algorithm}
We show that \textsc{Cycle Contraction} can be solved in $\mathcal{O}^*((2 + \epsilon_{\ell})^k)$ time where $0 < \epsilon_{\ell}  \leq  0.5509$ is inversely proportional to $\ell = n - k$ using the algorithm for \textsc{Path Contraction With Constrained Ends} as a subroutine.  

Let $(G,k)$ be an instance of \textsc{Cycle Contraction}. If $k \geq n-3$, then $(G,k)$ is a \yes-instance if and only if $G$ is $k$-contractible to a triangle (which can be determined in polynomial time). Hence, we may assume that $k \leq n-4$ and so we wish to determine if $G$ is contractible to $C_\ell$ where $\ell \geq 4$. We begin with two basic observations. 

\begin{observation}
\label{obs:sumofws}
If $G$ is $k$-contractible to $C_{\ell}$ with witness structure $(W_1,\ldots,W_\ell)$, then $\sum_{i=1}^{\ell} (|W_i|-1) \leq k$.  
\end{observation}

\begin{observation}
\label{obs:group-of-small-witness-set}
If $G$ is $k$-contractible to $C_{\ell}$ where $\ell \geq 4$ with witness structure $(W_1,W_2,\ldots,W_\ell)$, then there is a set of three consecutive witness sets $W_i$, $W_{i+1}$, $W_{i+2}$ such that $|W_i|+ |W_{i+1}|+|W_{i+2}| \leq 3+ \frac{3k}{\ell}$. 
\end{observation}
\begin{proof}
Suppose there is no set of three consecutive witness sets with the desired property. Then, for each $i \in [\ell-2]$, we have $|W_i|+ |W_{i+1}|+|W_{i+2}| > 3+ \frac{3k}{\ell}$. Further, $|W_{\ell-1}|+ |W_{\ell}|+|W_{1}| > 3+ \frac{3k}{\ell}$ and $|W_{\ell}|+ |W_{1}|+|W_{2}| > 3+ \frac{3k}{\ell}$. Summing up these inequalities leads to $3 \sum_{i=1}^{\ell} |W_i| > 3\ell+ \ell \frac{3k}{\ell}$ which in turn leads to $\sum_{i=1}^{\ell} |W_i| > k + \ell$. This leads to a contradiction by Observation~\ref{obs:sumofws}. \qed
\end{proof}

Observation \ref{obs:sumofws} immediately leads to a $\mathcal{O}^*(2^k)$-time algorithm for determining if $G$ is $k$-contractible to $C_{4}$ or not. 

\begin{lemma}
\label{lem:c4algo}
There is an algorithm that determines if $G$ is $k$-contractible to $C_{4}$ or not in $\mathcal{O}^*(2^k)$ time. 
\end{lemma}
\begin{proof}
If $G$ has more than $k+4$ vertices, then by Observation \ref{obs:sumofws}, $G$ is not $k$-contractible to $C_{4}$. Subsequently, assume $|V(G)|\leq k+4$. We iterate over all possible colorings of $V(G)$ using two colors and for each coloring, we contract each of the connected components of the two monochromatic subgraphs into a single vertex and check if the resultant graph is $C_4$. If no coloring results in $C_4$, then $G$ is not $k$-contractible to $C_{4}$. The set of colorings that results in $C_4$ are the set of $C_4$-witness structures of $G$ each of which uses at most $k$ edge contractions. As there are $2^{k+4}$ such colorings, the overall running time is $\mathcal{O}^*(2^k)$. \qed
\end{proof}

Next, we prove the following lemma on enumerating connected sets with restrictions on its size and its neighbourhood's size. 

\begin{lemma} \label{lem:ab-conn-sets-nbrs} 
Given integers $\alpha,\beta \in \mathbb{N}$, the number of triples $(X,Z,Y)$ where $Z$ is a $(\alpha,\beta)$-connected set and $X \uplus Y = N(Z)$ is $3^{\alpha + \beta} \cdot n$ and the set of all such triples can be enumerated in $\mathcal{O}^*(3^{\alpha + \beta})$ time.
\end{lemma}
\begin{proof}
We use a simple depth-bounded search tree algorithm for enumerating the required triples. Each node of the search tree is labelled as $(X,Z,Y)$ where $|Z| \leq \alpha$ and $|X \cup Y| \leq \beta$. A node $(X,Z,Y)$ is associated with a measure $\mu=(\alpha+\beta)-|Z|-|X \cup Y|$. 

Guess a vertex $v \in Z$. The root of the search tree is $(X=\emptyset,Z=\{v\},Y=\emptyset)$ with measure $\alpha-1+\beta$. At a node $(X,Z,Y)$, we choose an arbitrary vertex $u \in N(Z) \setminus (X \cup Y)$ and branch into three possibilities: $u \in Z$ or $u \in X$ or $u \in Y$ resulting in three new nodes. This branching is clearly exhaustive and the measure drops by one in each of the three branches. The depth of the tree is at most $\alpha+\beta$ and the number of leaves is at most $3^{\alpha + \beta}$. The leaves $(X,Z,Y)$ where $X \cup Y = N(Z)$, $|Z|=\alpha$, $|X \cup Y|=\beta$ correspond to the required triples such that $v \in Z$. 

Then, we run this branching algorithm for all choices of $v$. Thus, the total number of triples $(X,Z,Y)$ where $Z$ is a $(\alpha,\beta)$-connected set and $X \uplus Y = N(Z)$ is at most $3^{\alpha + \beta} \cdot n$.
Further, the set of all such triples can be enumerated in $\mathcal{O}^*(3^{\alpha + \beta})$ time.\qed
\end{proof}

Now, we have the following result. 

\fptalgocycle*
\begin{proof}
Consider an instance $(G,k)$ of \textsc{Cycle Contraction}. Let $\ell = n - k$. If $\ell = 3$, then we can determine if $G$ is $k$-contractible to $C_3$ or not in polynomial time. If $\ell = 4$, then we can determine if $G$ is $k$-contractible to $C_4$ or not using Lemma \ref{lem:c4algo} in $\mathcal{O}^*(2^k)$ time. Subsequently, we may assume that $\ell \geq 5$. Suppose $G$ is $k$-contractible to $C_{\ell}$ with witness structure $(W_1,W_2,\ldots,W_\ell)$. From Observation \ref{obs:group-of-small-witness-set}, we know that there is a set of three consecutive witness whose total size is at most $3+ \frac{3k}{\ell}$. Without loss of generality, we may assume that these sets are $W_1, W_2$ and $W_3$. Observe that $W_1, W_2$ and $W_3$ are connected sets with $N(W_2) \subseteq W_1 \cup W_2$. Let $X = W_1 \cap N(W_2)$ and $Y = W_3 \cap N(W_2)$. Once $W_2$, $X$ and $Y$ is known, the problem reduces to solving a \textsc{Path Contraction With Constrained Ends} instance of determining if $G-W_2$ is $k'$-contractible to a path with witness structure $(W'_1,W'_2,\ldots,W'_q)$ where $X \subseteq W'_1$ and $Y \subseteq W'_q$ where $k'=k-(|W_2|-1)$. Such a path witness structure $(W'_1,W'_2,\ldots,W'_q)$ along with $W_2$ is a cycle witness structure of $G$. 

An important subroutine in this algorithm is an efficient enumeration of $(X,W_2,Y)$ given by Lemma \ref{lem:ab-conn-sets-nbrs}. For each choice of $(X,W_2,Y)$, the corresponding \textsc{Path Contraction With Constrained Ends} problem on $G-W_2$ can be solved in $\mathcal{O}^*(2^{k'-|X|-|Y|})$ time where $k'=k-(|W_2|-1)$ by Theorem \ref{thm:path-ends-fpt-algo}. The number of choices of $(X,W_2,Y)$ with $|W_2|=\alpha$, $|X|+|Y| =\beta$ where $\alpha+\beta \leq 3+ \frac{3k}{\ell}$ is $\mathcal{O}^*(3^{\alpha+\beta})$ by Lemma \ref{lem:ab-conn-sets-nbrs}. For each such choice of $(X,W_2,Y)$, the corresponding \textsc{Path Contraction With Constrained Ends} problem on $G-W_2$ can be solved in $\mathcal{O}^*(2^{k'-|X|-|Y|})$ time where $k'=k-(|W_2|-1)$ by Theorem \ref{thm:path-ends-fpt-algo}. 

The total running time is $\mathcal{O}^*(3^{\frac{3k}{\ell}} \cdot 2^{k-{\frac{3k}{\ell}}})$ which is $\mathcal{O}^*((1.5^{c} \cdot 2)^k)$ where $0 \leq c \leq \frac{3}{\ell}$ and $\ell \geq 5$. Thus, the algorithm for \textsc{Cycle Contraction} runs in $\mathcal{O}^*((2 + \epsilon_{\ell})^k)$ time where $0 < \epsilon_{\ell}  \leq  0.5509$ and $\epsilon_\ell$ is inversely proportional to $\ell = n - k$. \qed
\end{proof}

\subsection{An Exponential-time Algorithm}
Now, we address the optimization version of \textsc{Cycle Contraction} where the objective is contracting the input graph $G$ to the largest cycle possible. By Theorem~\ref{thm:fpt-algo-cycle}, \textsc{Cycle Contraction} can be solved in $\mathcal{O}^*(2.5509^k)$ time. Since $1\leq k \leq n-1$, by invoking this algorithm for different values of $k$, we can determine the largest cycle to which $G$ can be contracted to in $\mathcal{O}^*(2.5509^n)$ time. In this section, we describe a faster algorithm for this problem. 

We begin with the following observation on contracting a graph to the largest even cycle. This is similar to Lemma~\ref{lem:c4algo}. 

\begin{lemma}
\label{lem:even-cycle-exact}
There is an algorithm that takes as input a graph $G$ on $n$ vertices and determines the longest even cycle $C_{\ell}$ to which $G$ is contractible to in $\mathcal{O}^*(2^n)$ time. Further, the set of all $C_{\ell}$-witness structures of $G$ can be enumerated in $\mathcal{O}^*(2^n)$ time.
\end{lemma}
\begin{proof}
We iterate over all possible colorings of $V(G)$ using two colors and for each coloring, we contract each of the connected components of the two monochromatic subgraphs into a single vertex and check if the resultant graph is an even cycle. The set of colorings that results in a largest even cycle is the required set of witness structures of $G$. As there are $2^n$ such colorings, the overall running time is $\mathcal{O}^*(2^n)$. \qed
\end{proof}

Observe that if $C_\ell$ is the largest even cycle to which $G$ can be contracted to, then the largest  cycle to which $G$ can be contracted to is either $C_\ell$ or $C_{\ell+1}$. Without loss of generality, assume $\ell \geq 4$. Suppose $G$ is contractible to $C_{\ell+1}$ with $(W_1,W_2,W_3,\ldots,W_\ell,W_{\ell+1})$ being the corresponding witness structure. Then, observe that by contracting one more edge with one endpoint in $W_i$ and the other endpoint in $W_{i+1}$, we get a $C_\ell$-witness structure $(W_1, W_2,W_3,\ldots,W_i \cup W_{i+1},\ldots,W_\ell,W_{\ell+1})$ of $G$. Therefore, one might be tempted to obtain a $C_{\ell+1}$-witness structure by using $C_{\ell}$-witness structures obtained from Lemma~\ref{lem:even-cycle-exact}. There are two challenges in this approach. First, we do not know which is the ``correct'' $C_{\ell}$-witness structure that leads to a $C_{\ell+1}$-witness structure. Second, even if we know the correct $C_{\ell}$-witness structure, we do not know which edge to ``uncontract'' in order to get a $C_{\ell+1}$-witness structure. The first challenge can be handled by enumerating all $C_{\ell}$-witness structures using Lemma~\ref{lem:even-cycle-exact} and checking each of them one by one. To handle the second challenge, we observe the following.

\begin{observation}
\label{obs:two-small-witness-set}
If $G$ is contractible to the odd cycle $C_{\ell+1}$ where $\ell \geq 4$ with witness structure $(W_1,W_2,\ldots,W_{\ell+1})$, then is a set of two consecutive witness sets $W_i$, $W_{i+1}$ such that $|W_i|+ |W_{i+1}| \leq \frac{2n}{\ell+1}$. 
\end{observation}
\begin{proof}
Assume on the contrary that there is no set of two consecutive witness sets with the desired property. Then, for each $i \in [\ell]$, we have $|W_i|+ |W_{i+1}| > \frac{2n}{\ell+1}$. Further, $|W_{\ell+1}|+ |W_{1}| > \frac{2n}{\ell+1}$. Summing up these inequalities leads to $2 \sum_{i=1}^{\ell+1} |W_i| > 2n$ which leads to a contradiction. \qed
\end{proof}

Now, once we have a $C_\ell$-witness structure $(W_1,W_2,\ldots,W_{\ell})$ of $G$, we need to go over each witness set $W_i$ whose size is at most $\frac{2n}{\ell+1}$ with the interpretation that $W_i$ is to be partitioned into two connected sets resulting in a $C_{\ell+1}$-witness structure of $G$. Consider such a candidate set $W_i$. Let $P=N(W_{i-1}) \cap W_i$ and $Q=N(W_{i+1}) \cap W_i$. Now, the task is to determine a partition $W'_i \uplus W''_i$ of $W_i$ with $P \subseteq W'_i$ and $Q \subseteq W''_i$ such that $W'_i$ and $W''_i$ are connected. This is known as the \textsc{2-Disjoint Connected Subgraphs} problem and is known to admit an algorithm with $\mathcal{O}^*(1.7804^{|W_i|})$ running time \cite{TelleV13}. Such a solution partition guarantees that $W'_i$ and $W''_i$ are adjacent as $W'_i \cup W''_i$ is connected. Then, $(W_1, W_2,W_3,\ldots,W_{i-1},W'_i,W''_{i}, W_{i+1},\ldots,W_\ell)$ is a $C_{\ell+1}$-witness structure of $G$ ascertaining that $G$ is contractible to $C_{\ell+1}$. If no $C_\ell$-witness structure  of $G$ leads to a $C_{\ell+1}$-witness structure, then we conclude that $G$ is not contractible to $C_{\ell+1}$. 

The overall running time of this algorithm is $\mathcal{O}^*(2^n 1.7804^{2n/5})$ and we have the following result.

\begin{theorem}
\label{thm:exact-algo-cycle}
\textsc{Cycle Contraction} can be solved in $\mathcal{O}^*(2.5191^n)$ time.
\end{theorem}


\section{Planar Graphs and Chordal Graphs} 
In this section, we discuss the complexity of \textsc{Path Contraction} and \textsc{Cycle Contraction} on planar graphs and chordal graphs. 

\subsection{Planar Graphs: Path Contraction is Polynomial-time Solvable}
One of the most important special graph classes in the realm of contraction problems is the class of planar graphs. {\em Planar graphs} are graphs that can be drawn on the plane such that no pair of edges intersect except possibly at their endpoints. Such a drawing is called a {\em planar embedding}. Edge contractions play an important role in planar graph theory and planarity is preserved under edge contractions. Hammack~\cite{Hammack99} gave a polynomial-time algorithm that determines the largest cycle to which the given planar graph is contractible to. To the best of our knowledge, the status of \textsc{Path Contraction} in planar graphs is open. In this section, we give a polynomial-time algorithm for this problem that uses the \textsc{Cycle Contraction} algorithm of \cite{Hammack99} as a subroutine.

\pathplanar*
\begin{proof}
Consider an instance $(G,k)$ of \textsc{Path Contraction} where $G$ is a connected planar graph with a planar embedding fixed. If $|V(G)| < k+2$, $(G,k)$ is a trivial \yes-instance. Subsequently, assume $|V(G)| \geq k+2$. Suppose $G$ is $k$-contractible to a path with witness structure $(W_1, W_2, \ldots, W_q)$. Then, there are vertices $x \in W_1$ and $y \in W_q$ that are in the same face of the embedding. Construct a graph $G_{x,y}$ by adding a path $(z_1,\ldots,z_{k+1})$ of $k+1$ new vertices to $G$ such that $z_1$ is adjacent to $x$ and $z_{k+1}$ is adjacent to $y$. Observe that $G_{x,y}$ is also planar. We claim that $G$ is $k$-contractible to a path with witness structure $(W_1, W_2, \ldots, W_q)$ such that $x \in W_1$ and $y \in W_q$ if and only if $G_{x,y}$ is $k$-contractible to a cycle.  

Suppose $G$ is $k$-contractible to a path with witness structure $(W_1, W_2, \ldots, W_q)$ such that $x \in W_1$ and $y \in W_q$. Then, $(W_1, W_2, \ldots, W_q, \{z_{k+1}\},\{z_k\},\ldots,\{z_1\})$ is a cycle witness structure of $G_{x,y}$ in which at most $k$ edges are contracted. In other words, $G_{x,y}$ is $k$-contractible to a cycle. Conversely, suppose $G_{x,y}$ is $k$-contractible to a cycle with witness structure $(W_1, W_2, \ldots, W_q)$. Let $Z=\{z_1,\ldots,z_{k+1}\}$. Let us first consider the case when $x$ and $y$ are in the same witness set. Without loss of generality, assume that $x,y \in W_1$. If $Z \subseteq W_1$, we get a contradiction since more than $k$ edges are contracted. Therefore, we may assume that $Z \cap W_2 \neq \emptyset$. Let $W_i$ denote a witness set containing a vertex of $V(G_{x,y}) \setminus (Z \cup \{x,y\})$ for $i>1$. Such a set exists as $|V(G)| \geq k+2$. Since $N(Z)=\{x,y\}$, it follows that there is no path from $W_2$ to $W_i$ which leads to a contradiction. Therefore, we may conclude that $x$ and $y$ are in different witness sets. Let $x \in W_1$ and $y \in W_p$ for some $1<p \leq q$. Since $N(Z)=\{x,y\}$, all vertices of $Z$ are in $\bigcup_{i=1}^{p}W_i$ or in $\bigcup_{i=p}^{q}W_i \cup W_1$. Assuming the latter case without loss of generality, it follows that $(W_1 \setminus Z, W_2, \ldots, W_p \setminus Z)$ is a path witness structure of $G$ where at most $k$ edges are contracted. That is, $G$ is $k$-contractible to a path.

Now, in order to determine if $G$ is $k$-contractible to a path, we go over all possible choices of $x$ and $y$, construct $G_{x,y}$, use the algorithm of Hammack~\cite{Hammack99} to determing if $G_{x,y}$ is $k$-contractible to a cycle or not which in turn by determines if $G$ is $k$-contractible to a path or not. As the number of choices for $x,y$ is at most the number of pairs of vertices of $G$ that are in the same face, we obtain a polynomial-time algorithm for \textsc{Path Contraction} on planar graphs.\qed
\end{proof}

\subsection{Chordal Graphs: Fine-grained Complexity of Path Contraction}
Now, we move on to chordal graphs. Chordal graphs are graphs in which every induced cycle is a triangle. Chordality is a hereditary property and edge contractions preserve chordality. The class of chordal graphs subsumes many popular graph classes like complete graphs, trees, $k$-trees, interval graphs and split graphs. Chordal graphs can be recognized in linear time and are amenable to linear time algorithms for several problems that are \NP-complete in general like \textsc{Vertex Cover}, \textsc{Vertex Coloring}, \textsc{Clique Cover} and \textsc{Clique}. A well-known characterization of chordal graphs is via {\em clique trees}. A {\em clique tree} of a connected graph $G$ is a tree that has as vertices the maximal cliques of $G$ and has edges such that each subgraph induced by those cliques that contain a particular vertex of $G$ is a subtree. If a graph $G$ has a clique tree $T$, then $T$ is also a tree decomposition of $G$ with bags as the maximal cliques of $G$ and the treewidth of $G$ is the number of vertices in a maximum clique minus one\footnote[2]{The treewidth of a graph measures how close the graph is to a tree. See \cite{CyganFKLMPPS15} for the definitions of tree decomposition and treewidth}. It is well-known that $G$ is chordal if and only if $G$ has a clique tree~\cite{Gavril74}. Further, a clique tree of a chordal graph can be computed in linear time. Moreover, the number of maximal cliques of a chordal graph on $n$ vertices is at most $n$. For further details, we refer the reader to the book by Golumbic~\cite{Golumbic80}.  

It is easy to verify that \textsc{Cycle Contraction} is linear-time solvable on chordal graphs -- a chordal graph is contractible to a cycle if and only if it is contractible to a triangle. Heggernes et al.~\cite{HeggernesHLP14} proved that \textsc{Path Contraction} is polynomial-time solvable on chordal graphs by giving an $\mathcal{O}(nm)$-time algorithm where $m$ is the number of edges in the input graph. This algorithm is an $\mathcal{O}(n^2 \cdot tw)$-time algorithm where $tw$ is the treewidth of $G$. We give a fine-grained reduction from \textsc{Orthogonal Vectors} to \textsc{Path Contraction} on chordal graphs to show that the latter problem has no $\mathcal{O}(n^{2-\epsilon} \cdot 2^{o(tw)})$-time algorithms under the Orthogonal Vectors Conjecture. In the \textsc{Orthogonal Vectors} problem, given two sets of $n$ $d$-dimensional boolean vectors $X=\{x_1, x_2, \ldots, x_n\}$ and $Y=\{y_1, y_2, \ldots, y_n\}$, the objective is to determine if there is a pair of vectors $x_i \in X$ and $y_j \in Y$ that are orthogonal, i.e., $\langle x_i \cdot y_j \rangle= 0$. It is easy to see that \textsc{Orthogonal Vectors} can be solved in $\calO(n^2 d)$ time. The Orthogonal Vectors Conjecture~\cite{WilliamsY14} states that \textsc{Orthogonal Vectors} cannot be solved in $\calO(n^{2-\epsilon})$ time for any $\epsilon > 0$ when $d=\mathcal{O}(\log n)$.



%

\chordallb*
\begin{proof}
   We first give a polynomial-time reduction from \textsc{Orthogonal Vectors} to \textsc{Path contraction} on chordal graphs. Consider an instance $(X',Y',n,d)$ of \textsc{Orthogonal Vectors} where $d=\mathcal{O}(\log n)$.  We construct a graph $G$ as follows. 
    \begin{itemize}
        \item For every $x \in X'$, add a vertex $x$ and for every $y \in Y'$, add a vertex $y$. Let $X=\{x_1,\ldots,x_n\}$ and $Y=\{y_1,\ldots,y_n\}$. 
        \item Add a clique $Z=\{z_X,z_1,z_2, \ldots, z_d,z_Y\}$ such that $z_X$ is adjacent to every vertex in $X$ and $z_Y$ is adjacent to every vertex in $Y$. 
        \item For each $x \in X'$ and $i \in [d]$ with $x(i)=1$, add an edge between $x$ and $z_i$. 
        \item For each $y \in Y'$ and $i \in [d]$ with $y(i)=1$, add an edge between $y$ and $z_i$. 
    \end{itemize}
    See Figure~\ref{fig:ov-path} for an illustration. 
    \begin{figure}[ht]
    \centering
    \includegraphics[scale=.35]{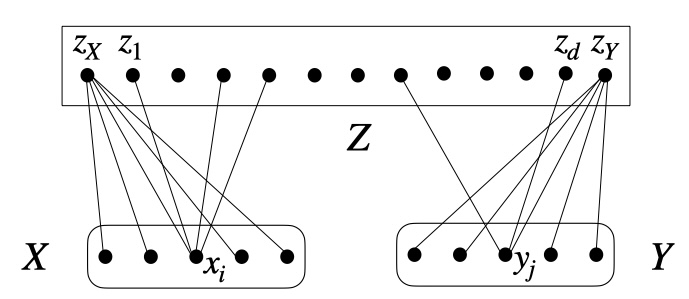}
    \caption{The (chordal) graph $G$ in the reduction from \textsc{Orthogonal Vectors} to \textsc{Path Contraction} where $X \cup Y$ is an independent set and $Z$ is a clique.}
     \label{fig:ov-path}
    \end{figure}
    Observe that $V(G)$ can be partitioned into $X \cup Y$ and $Z$ where $X \cup Y$ is an independent set and $Z$ is a clique. Therefore, $G$ is a split graph which is in turn chordal. Note that the diameter of $G$ is at most 3. It is easy to verify that the reduction takes $\mathcal{O}(nd+d^2)$ time and $G$ is connected.  
    
    We claim that there is a pair of orthogonal vectors $x_i \in X'$ and $y_j \in Y'$ if and only if $G$ is $k$-contractible to a path where $k=2n+d+2-4$. Suppose $\langle x_i \cdot y_j \rangle= 0$. Then, there is no index $\ell$ such that $x_i(\ell)=y_j(\ell)=1$. Let $Z_1 =N(x_i)$, $Z_2 =N(y_j)$ and $Z_3=Z \setminus (Z_1 \cup Z_2)$. Define four sets $W_1, W_2, W_3, W_4$ where $W_1=\{x_i\}$, $W_4=\{y_j\}$, $W_2=Z_1 \cup Z_3 \cup (X \setminus \{x_i\})$ and $W_3=Z_2 \cup (Y \setminus \{y_j\})$. By construction of $G$ and due to the fact that $<x_i \cdot y_j >= 0$, $W_1, W_2, W_3, W_4$ form a partition of $V(G)$. Also, $E(W_1,W_4)= \emptyset$, $E(W_1,W_3)= \emptyset$ and $E(W_2,W_4)= \emptyset$. The sets $W_1$ and $W_4$ are singletons and $z_X \in W_2$ and $z_Y \in W_3$. Thus, each $W_i$ is a connected set for $i\in[4]$. Further, $E(W_1,W_2)$ and $E(W_3,W_4)$ are non-empty since $x_i$ and $y_j$ are not the zero vector. Moreover, $E(W_2,W_3)$ is non-empty since $z_X \in W_2$ and $z_Y \in W_3$. Therefore, $(W_1,W_2,W_3,W_4)$ is a $P_4$-witness structure of $G$. In other words, $G$ is $k$-contractible to a path. 

    Conversely, suppose $G$ is $k$-contractible to a path. Then such a path should have exactly 4 vertices since $dia(G) \leq 3$, $k=2n+d+2-4$ and $|V(G)|=2n+d+2$. Let $(W_1, W_2, W_3, W_4)$ be a $P_4$-witness structure of $G$. As $Z$ is a clique, the vertices in $Z$ are in at most two (consecutive) witness sets. If $Z \subseteq W_1 \cup W_2$, then $W_4=\emptyset$ since $V(G) \setminus Z$ is an independent set. Similarly, if $Z \subseteq W_3 \cup W_4$, then $W_1=\emptyset$. Therefore, $Z \subseteq W_2 \cup W_3$. This implies that $(X \cup Y) \cap W_1 \neq \emptyset$ and $(X \cup Y) \cap W_4 \neq \emptyset$. Further, if $X \cap W_1 \neq \emptyset$ (implying that $z_X \in W_2$), then $X \cap W_4 = \emptyset$. Similarly, if $Y \cap W_1 \neq \emptyset$ (implying that $z_Y \in W_2$), then $Y \cap W_4 = \emptyset$. Therefore, without loss of generality, we may assume that $X \cap W_4 = \emptyset$ and $Y \cap W_1 = \emptyset$. Then, $W_1$ and $W_4$ are singletons as $N(X) \subseteq Z$ and $N(Y) \subseteq Z$. Let $W_1=\{x_i\}$ and $W_2=\{y_j\}$. We claim that $\langle x_i \cdot y_j \rangle= 0$. If there is an index $\ell$ such that $x_i(\ell)=y_j(\ell)=1$, then $x_i$ and $y_j$ are adjacent to $z_\ell$. This results in $E(W_1,W_3)\neq \emptyset$ or $E(W_2,W_4)\neq \emptyset$ which leads to a contradiction.

    Suppose there is an algorithm for \textsc{Path Contraction} on chordal graphs that runs in $\mathcal{O}(n^{2-\epsilon} \cdot 2^{o(tw)})$ time for some $\epsilon >0$. 
    Given an instance $(X,Y,n,d)$ of \textsc{Orthogonal Vectors} where $d=\mathcal{O}(\log n)$, we first construct the chordal graph $G$ as detailed above in $\mathcal{O}(nd+d^2)$ time and then determine if $G$ can be contracted to a $P_4$ in $\mathcal{O}(n^{2-\epsilon} \cdot 2^{o(tw)})$. As the treewidth of $G$ is $d+2$ which is $\mathcal{O}(\log n)$, this algorithm runs in time $\mathcal{O}(n^{2-\epsilon'} \cdot 2^{o(\log n)})$ which is $\mathcal{O}(n^{2-\epsilon'} \cdot n^{o(1)})$ for some $\epsilon' >0$. Thus, we solve the instance $(X,Y,n,d)$ of \textsc{Orthogonal Vectors} in $\mathcal{O}(n^{2-\delta})$ time for some $\delta >0$. This contradicts the Orthogonal Vectors Conjecture.\qed
\end{proof}

This result shows that the $\mathcal{O}(n^2 \cdot tw)$-time algorithm is optimal -- if one wishes to reduce the quadratic dependency on $n$, an exponential overhead in terms of treewidth is inevitable. We remark that this fine-grained complexity analysis is non-trivial only when the input graph is sparse; if $m = \mathcal{O}(n^2)$, then we need to spend $\Omega(n^2)$ time just to read the graph.

\section{Concluding Remarks}
In this paper, we revisited \textsc{Path Contraction} and \textsc{Cycle Contraction} from two perspectives -- the first one focusing on the parameterized complexity and the second one on polynomial-time solvability in special graph classes. We gave improved \FPT\ algorithms for both the problems using a dynamic programming algorithm that solves \textsc{Path Contraction With Constrained Ends} as a subroutine. Then, we showed how to solve \textsc{Path Contraction} using \textsc{Cycle Contraction} as a subroutine with only a polynomial-time overhead. This led to a polynomial-time algorithm for \textsc{Path Contraction} in planar graphs using a known polynomial-time algorithm for \textsc{Cycle Contraction} in planar graphs as a subroutine. Finally, we gave a fine-grained complexity analysis of \textsc{Path Contraction} in chordal graphs. 

Just as breaking the $\mathcal{O}^*(2^n)$ barrier proved challenging for \textsc{Path Contraction}, we believe that beating the $\mathcal{O}^*(2^k)$ bound would be interesting. Similarly, coming up with an $\mathcal{O}^*(2^k)$-time algorithm for \textsc{Cycle Contraction} is a natural future direction. As such an algorithm would result in solving \textsc{Cycle Contraction} in $\mathcal{O}^*(2^n)$ time, coming up with an $\mathcal{O}^*(2^n)$-time algorithm for contracting to $C_5$ maybe a first step towards this goal. Recall that our algorithm for \textsc{Cycle Contraction} has the worst case running time when the target graph is $C_5$. Finally, determining the longest cycle to which an $H$-free graph (for a fixed $H$) is contractible is another interesting future direction. A similar study on $H$-free graphs in the context of longest paths is known~\cite{KernP20}. Note that assuming \P$\neq$\NP, the complexities of contracting to a longest path and longest cycle do not coincide on $H$-free graphs.\\

\smallskip \noindent \textbf{Acknowledgement. }We thank Roohani Sharma for initial discussions.

%
%
%

\bibliographystyle{splncs04}  
\bibliography{references}

\end{document}